\documentclass[
twocolumn,
showpacs,preprintnumbers,amsmath,amssymb,aps,pra
]{revtex4-1}

\usepackage{graphicx}% Include figure files
\usepackage{bm}% bold math
\usepackage{amsfonts}
\usepackage{amscd}
\usepackage{amsmath}    % need for subequations
\usepackage{enumerate}

\usepackage{amssymb}
\usepackage{amsthm}

\usepackage[breaklinks]{hyperref}

\newtheorem{theorem}{Theorem}
\newtheorem{lemma}{Lemma}
\newtheorem{corollary}{Corollary}

\newcommand{\bra}[1]{\mbox{$\left\langle #1 \right|$}}
\newcommand{\ket}[1]{\mbox{$\left| #1 \right\rangle$}}
\newcommand{\braket}[2]{\mbox{$\left\langle #1 | #2 \right\rangle$}}

\newcommand{\braa}[1]{\mbox{$\langle #1 |$}}
\newcommand{\kett}[1]{\mbox{$| #1 \rangle$}}
\newcommand{\braakett}[2]{\mbox{$\langle #1 | #2 \rangle$}}

\def\tr{{\rm Tr}}

\def\IC{{\mathbb C}}

\newcommand{\maxnorm}[1]{\left\lVert #1 \right\rVert_{\text{max}}}

\def\KrausDim{{K}}
\def\ChannelDim{{n}}
\def\ChannelDimO{{m}}
\DeclareMathOperator{\myUgrp}{U}

\def\myrm{\rm}

\def\dmathX#1#2{
$$\lineskiplimit=1000pt \advance\lineskip by #1\jot 
\mathsurround=0pt \tabskip=0pt plus 1000pt
\everycr{\noalign{\penalty\interdisplaylinepenalty}}
\halign to \displaywidth{
\hfil$\displaystyle{##}$\tabskip=0pt&%
\hfil $\displaystyle{{}##{}}$\hfil &%
\hfil $\displaystyle{{}##{}}$\hfil &%
$\displaystyle{##}$\hfil \tabskip=0pt plus 1000pt minus 1000pt&%
\refstepcounter{equation}\label{##}\llap{(\theequation)}\tabskip=0pt\cr
\noalign{\ifdim \prevdepth>-1000pt \vskip -#1\jot\fi}
#2\crcr}$$}

\newcommand\undermat[2]{%
  \makebox[0pt][l]{$\smash{\underbrace{\phantom{%
    \begin{matrix}#2\end{matrix}}}_{\text{$#1$}}}$}#2}

\begin{document}

\title
{Time-Energy Costs of Quantum Measurements}

\author{Chi-Hang Fred Fung}
\email{chffung@hku.hk}
\affiliation{Department of Physics and Center of Theoretical and Computational Physics, University of Hong Kong, Pokfulam Road, Hong Kong}
\author{H.~F. Chau}
\affiliation{Department of Physics and Center of Theoretical and Computational Physics, University of Hong Kong, Pokfulam Road, Hong Kong}

\begin{abstract}
Time and energy of quantum processes are a tradeoff against each other.
We propose to ascribe to any given quantum process a time-energy cost to quantify how much computation it performs.
Here, we analyze the time-energy costs for general quantum measurements, along a similar line as our previous work for quantum channels, and prove exact and lower bound formulae for the costs.
We use these formulae to evaluate the efficiencies of actual measurement implementations.
We find that one implementation for a Bell measurement is optimal in time-energy.
We also analyze the time-energy cost for unambiguous state discrimination and find evidence that only a finite time-energy cost is needed to distinguish any number of states.
\end{abstract}


\maketitle

\section{Introduction}

Quantum mechanical systems cannot evolve with an arbitrary speed and an arbitrary energy.
The evolution speed and system energy are constrained by time-energy uncertainty relations (TEURs)~\cite{Lloyd2000}.
Many TEURs have been proposed and investigated~\cite{Mandelstam1945,
Bhattacharyya1983,Anandan1990,Uhlmann1992,Vaidman1992,Pfeifer1993,Margolus1996,Margolus1998,Chau2010,
Giovannetti2003,Giovannetti2003b,Zander2007,
Taddei2013,delCampo2013} and they follow a general form in which the product of the evolution time (needed to evolve the initial state to the final state) and the system energy (or a function of the eigen-energies) is upper bounded by some number dependent on the closeness between the initial and final states.
Recognizing that time and energy are a tradeoff against each other, we proposed to regard time energy as a single measure for the resource consumed by a quantum process~\cite{Chau2011,Fung:2013:Time-energy}.
Essentially, a high time-energy cost indicates that the process 
requires a long time to complete at a low system energy level or a high system energy level for a short completion time.
We motivated definitions for the time-energy measures for unitary transformations~\cite{Chau2011} and quantum channels~\cite{Fung:2013:Time-energy} by a TEUR proved earlier~\cite{Chau2010}.
In this work, we investigate the time-energy measure for general quantum measurements also called positive operator-valued measures (POVM).
Quantum measurements are quantum evolutions of some quantum states that eventually produce classical outputs (i.e., by triggering a detector).
Thus, quantum measurements are also restrained by TEURs and the concept of time-energy cost also applies to them.
Essentially, ``easy'' measurements (e.g., directly detecting the input states) would incur small time-energy costs.
More specifically, a quantum measurement can be considered as a unitary operation in a larger Hilbert space containing the system to be measured and an ancillary system indicating the measurement outcome.
We define the time-energy cost of a measurement as the time-energy cost for this unitary operation which we have already quantified before~\cite{Chau2011,Fung:2013:Time-energy}.

The time-energy cost of a measurement given the POVM description may be used to judge the efficiency of an actual implementation.
The time-energy cost of an implementation can be computed based on the actual experimental components (such as beam splitters) used and the time-energy cost of the POVM can be computed (or bounded) using the results of this work.
A small difference between these cost values indicates that the actual implementation is quite efficient already, consuming close to the fundamental minimal time and energy to run.

In this work, 
we derive lower bounds on the time-energy cost of POVM and obtain the exact value for the time-energy cost in some special cases.
These results are applied to some examples.
In particular,
we compute the time-energy costs of linear optics based implementations of Bell measurements and a POVM with rank-2 elements, and compare them with the ideal time-energy costs given the POVM descriptions.
We find that the Bell measurement implementation that projects onto one Bell state is optimal, but that projects onto two Bell states is not.
Also, our calculation indicates that the implementation of the POVM with rank-2 elements may be far from optimal.
In addition, we study the time-energy cost for the optimal unambiguous state discrimination (USD) for distinguishing symmetric coherent states.
Interestingly, the cost lower bound increases but saturates to some value as the number of states increases.
This may indicate that a finite time-energy resource is enough to distinguish any number of states.

We motivate a time-energy measure based on the following TEUR by Chau~\cite{Chau2010}.
Given a time-independent Hamiltonian $H$ of a system, 
the time 
$t$
needed to evolve a state $\ket{\Phi}$ under the action of $H$ to a state 
whose fidelity~\footnote{We adopt the fidelity definition $F(\rho,\sigma)=\big(\tr \sqrt{\rho^{1/2} \sigma \rho^{1/2}}\big)^2$ for two quantum states $\rho$ and $\sigma$.} is less than or equal to $\epsilon$ satisfies the TEUR
\begin{align}
\label{eqn-time-energy-relation1}
t
\geq \frac{(1-\sqrt{\epsilon})\hbar}{A \sum_j |\alpha_j|^2 |E_j|}
\end{align}
where $E_j$'s are the eigenvalues of $H$ with the corresponding normalized energy eigenvectors $\ket{E_j}$'s, $\ket{\Phi}=\sum_j \alpha_j \ket{E_j}$, and $A \approx 0.725$ is a universal constant.
Based on this equation,
a weighted sum of $|t E_j|$'s 
serves as
an indicator of the time-energy resource needed to perform $U=\exp(-i H t / \hbar)$.
Thus, this motivates the following definition
of the time-energy cost of a unitary matrix $U \in \myUgrp(r)$~\cite{Chau2011}:
\begin{align}
\label{eqn-definition-maxnorm-for-U}
\maxnorm{U}&=\max_{1 \le j \le r}  |\theta_j|
\end{align}
where $U$ has eigenvalues $\exp(-i E_j t/\hbar)\equiv \exp(\theta_j)$ for $j=1,\dots,r$ and $E_j$ are the eigenvalues of the Hamiltonian $H$~\footnote{We remark that our previous works~\cite{Chau2011,Fung:2013:Time-energy} consider more general measures by taking linear combinations of $|\theta_j|$'s.  Here, we only consider the maximum $|\theta_j|$.}.
We assume that all angles are taken in the range $(-\pi,\pi]$.

The concept of the time-energy cost has been extended to quantum channels by considering a unitary extension in a larger Hilbert space and regarding the cost of the unitary as the cost of the quantum channel~\cite{Fung:2013:Time-energy}.
The time-energy resource for a quantum channel $\mathcal F$ with Kraus operators $\{F_1,\ldots, F_K\}$ is defined as
\dmathX2{
\maxnorm{\mathcal{F}} &\equiv& \min_U & \maxnorm{U}  
&eqn-energy-measure-general-channel\cr
&&
\text{s.t.} &
\mathcal{F}(\rho) = \tr_B [ U_{BA} (\ket{0}_B\bra{0} \otimes  \rho_A) U_{BA}^\dag ]
\: \forall \rho.
\cr
}
where the channel $\mathcal{F}$ acts on state $\rho$ in system $A$ and the unitary extension $U_{BA}$ includes system $B$ prepared in a standard state.
In this definition, we seek the unitary extension that consumes the least time energy.
We previously found bounds on $\maxnorm{\mathcal{F}}$ for general channels and obtained the exact value of $\maxnorm{\mathcal{F}}$ for some special channels including the depolarizing channel~\cite{Fung:2013:Time-energy}.

In this paper, we consider the time-energy cost for general quantum measurements on finite-dimensional systems.
A POVM can be cast as a quantum channel, and thus our previous result~\cite{Fung:2013:Time-energy} may be applied.
However, since there are extra unitary degree of freedom on the POVM elements and freedom in the labelings of the detection events (more explanation later), more analysis is needed to reuse the previous result for quantum channels.

We remark that a similar work by Uzdin and Gat \cite{Uzdin:2013:time-energy} derives results for the time-energy cost for USD measurements with rank-1 projectors.
In this work, we derive results for the time-energy cost for general POVM.

The organization of this paper is as follows.
We first introduce some notations and review some existing results 
in Sec.~\ref{sec-preliminary}.
These results are used to prove formulae for the time-energy cost for POVM in Sec.~\ref{sec-TE-POVM}.
In Sec.~\ref{sec-examples}, we apply the lower bound and exact formulae for the POVM time-energy cost to a few examples.
Finally, we conclude in Sec.~\ref{sec-conclusion}.

\section{Preliminary}
\label{sec-preliminary}

Denote by $\myUgrp(r)$ the group of $r \times r$ unitary matrices.
Given a matrix $U$, its $(i,j)$ element is denoted by $U(i,j)$,
row $i$ by $U(i,*)$, and column $j$ by $U(*,j)$.
We adopt the convention that $\cos^{-1}$ always returns an angle in the range $[0,\pi]$.

The quantum channel $\mathcal F$ is described by
$$
{\mathcal F}(\rho) = \sum_{i=1}^\KrausDim F_i \rho F_i^\dag
$$
where the Kraus operators are $F_i \in \IC^{m \times n}$.
We assume without loss of generality that $m \ge n$, since we can zero pad the Kraus operators and extract the non-zero subspace of the channel output.
We only consider finite-dimensional systems, i.e., $m,n < \infty$.

Define a map from a sequence of Kraus operators 
$(F_1,F_2,\ldots, F_\KrausDim)$ 
to a $\KrausDim \ChannelDimO \times \ChannelDim$ matrix 
as follows:
\begin{align}
g(F_1,F_2,\dots,F_{\KrausDim}) 
&\triangleq
\begin{bmatrix}
F_1 
\\
F_2 
\\
\vdots
\\
F_{\KrausDim}
\end{bmatrix}
\in \IC^{\KrausDim \ChannelDimO \times \ChannelDim}.
\end{align}
Because $\sum_{j=1}^{\KrausDim} F_j^\dag F_j = I$, the columns of 
$g(F_1,F_2,\dots,F_{\KrausDim})$
are orthonormal and 
$g(F_1,F_2,\dots,F_{\KrausDim})$
can be regarded as a submatrix of a unitary one.

\subsection{
Partial $U$ problem}

Problem~\eqref{eqn-energy-measure-general-channel} defines the time-energy cost for a general quantum channel.
Note that two sets of Kraus operators $\{F_1,\ldots, F_K\}$ and $\{F_1',\ldots, F_K'\}$ represent the same quantum channel if and only if $F_i'=\sum_{j=1}^K w_{ij} F_j$ for all $i$ and for some unitary matrix $[w_{ij}]$ (see Ref.~\cite{Nielsen2000}).
Thus, to solve problem~\eqref{eqn-energy-measure-general-channel}, one needs to consider all possible Kraus representations.
Let us propose a simpler but related problem, which will be useful for analyzing the time-energy cost for POVM in Sec.~\ref{sec-TE-POVM}.
Consider the time-energy cost for a sequence of Kraus operators.
We define the partial $U$ problem for the submatrix $g(F_1,F_2,\dots,F_{\KrausDim})$ as
\dmathX2{
\maxnorm{ g(F_1,F_2,\dots,F_{\KrausDim}) } \equiv\hspace{-2.6cm}& \cr
&& \displaystyle\min_{U} & 
\maxnorm{ U }&eqn-problem-partial-with-g\cr
&&\text{s.t.}&
U=
\begin{bmatrix}
F_1 
& * & * & \cdots & *
\\
F_2 
& * & \ddots & & *
\\
\vdots & \vdots &&& \vdots
\\
\undermat{n}{F_{\KrausDim}}
& * & * & \cdots & *
\end{bmatrix}
\in \myUgrp(\KrausDim m).
\cr
\cr
}
Here, the first $n$ columns are fixed and the optimization is over the remaining $\KrausDim m-n$ columns.
We proved formulae that upper and lower bound this problem in Ref.~\cite{Fung:2013:Time-energy} and we summarize the results in Appendix~\ref{app-summary}.

Note that $g(F_1,F_2,\dots,F_{\KrausDim})$
has the following property:
\begin{lemma}
\label{lemma-U-norm-unchanged-by-conjugation-2}
{\myrm
\begin{align*}
&\maxnorm{ g(F_1,F_2,\dots,F_{\KrausDim}) }
\\
=
&
\maxnorm{ g(\hat{Q}F_1Q^\dag,F_2Q^\dag,\dots,F_{\KrausDim}Q^\dag) }
\end{align*}
for any unitary matrix $Q \in \myUgrp(n)$ and 
\begin{equation}
\label{eqn-U-norm-unchanged-by-conjugation-Qhat}
\hat{Q}=
\begin{bmatrix}
Q & 0
\\
0& \bf{1}
\end{bmatrix}
\in \myUgrp(m).
\end{equation}
}
\end{lemma}%
This lemma is Lemma~\ref{lemma-U-norm-unchanged-by-conjugation} in Appendix~\ref{app-summary} in another form.
This form facilitates our later analysis.

\section{Time-energy cost of POVM}
\label{sec-TE-POVM}

\begin{figure}[t!]
\begin{center}
\includegraphics[width=\columnwidth]{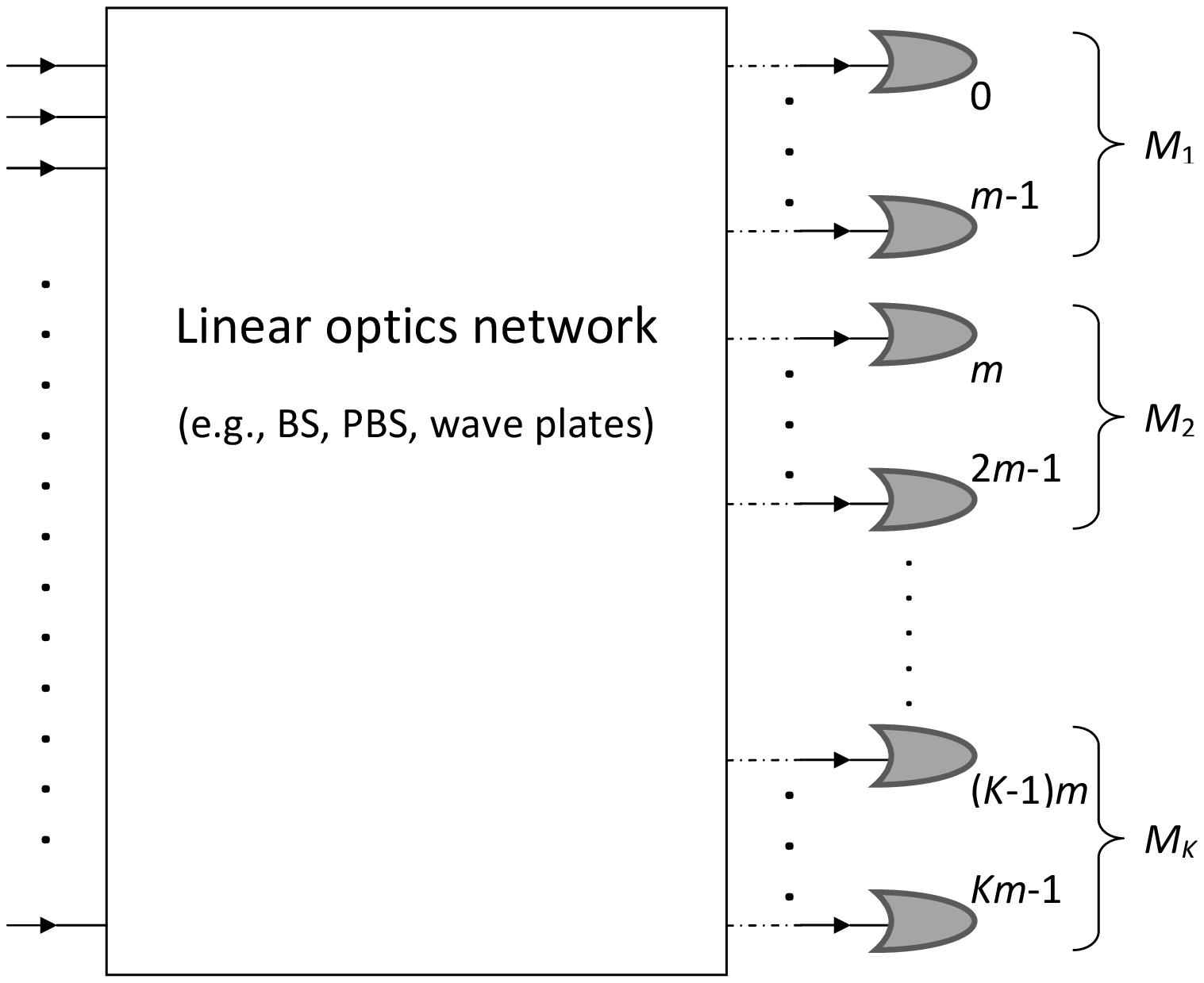}
\caption{
\label{fig-POVM}
Example implementation of a POVM based on linear optics.
In this example, the first $m$ detection events map to the first POVM element $M_1$, and the next $m$ detection events map to the second POVM element $M_2$, and so on.
}
\end{center}
\end{figure}

We are given a POVM $\mathcal M$ with elements $\{ M_i \in \IC^{n \times n}: i=1,\dots,K \}$ expressed in the basis $\{\ket{\overline{0}},\dots,\ket{\overline{n-1}}\}$, which, for example, may correspond to the input modes of beam splitters.
Note that $\sum_{i=1}^K M_i = I$ and $M_i$ is positive semidefinite.

An experiment implementing the POVM takes an input state in that basis and runs a quantum circuit to produce detection events corresponding to $\{M_i\}$.
We can label the detection events using another basis 
$\{\ket{0}, \dots, \ket{\KrausDim \ChannelDimO-1}\}$, which, for example, may correspond to the output modes of beam splitters.
Figure~\ref{fig-POVM} shows an example using linear optics to implement the POVM where each detection event corresponds to a detector click.
In the simplest case, the $\ChannelDimO$ detection events $\ket{(i-1) \ChannelDimO},\dots,\ket{i \ChannelDimO-1}$ map to $M_i$.
This corresponds to embedding the POVM in a unitary matrix $U$ in a larger space of dimension $\KrausDim \ChannelDimO$ and the projection onto detection event $\ket{j}$ indicates an outcome for $M_i$ according to the above mapping.
(We note that in reality, these projections need not be separately detected.)
This means that $U$ has to satisfy
\begin{align*}
\sum_{z=(i-1) \ChannelDimO}^{i \ChannelDimO-1} \bra{z} U \rho U^\dag \ket{z} &= \tr ( M_i \bar{\rho}) \hbox{ for all } i=1,\dots,K
\end{align*}
for any input state $\bar{\rho} \in \IC^{n \times n}$
and
\begin{align*}
\rho=
\begin{bmatrix}
\bar{\rho} & \bf{0}
\\
\bf{0} & \bf{0}
\end{bmatrix} 
\in \IC^{\KrausDim \ChannelDimO \times \KrausDim \ChannelDimO}
\end{align*}
is the input state in the larger space using basis $\{\ket{\bar{0}}, \dots, \ket{\overline{\KrausDim \ChannelDimO-1}}\}$.
Thus, $U$ is of the form 
\begin{align}
U=
\begin{bmatrix}
F_1 
& * & * & \cdots & *
\\
F_2 
& * & \ddots & & *
\\
\vdots & \vdots &&& \vdots
\\
\undermat{n}{F_{\KrausDim}}
& * & * & \cdots & *
\end{bmatrix}
\in \myUgrp(\KrausDim m)
\\ \nonumber
\end{align}
in which element $(i,j)$ corresponds to $\ket{i}\bra{\bar{j}}$,
and the Kraus operators are of the form
\begin{align}
\label{eqn-general-form-Kraus-operator}
F_i=V_i
\begin{bmatrix}
\sqrt{M_i} 
\\
\bf{0} 
\end{bmatrix} \in \IC^{m \times n},
\end{align}
where 
$V_i \in \myUgrp(m)$ that we may freely choose.
To maintain generality, we allow zeros to be padded in $F_i$.
In essence, the projections corresponding to the first $m$ rows of $U$ correspond to POVM outcome 1, and the next $m$ rows to POVM outcome 2, and so on.
These projections are the detection events when $U$ is directly implemented in an experiment and the order of them (i.e., the order of the rows of $U$) is meaningless.
In other words, we may arbitrarily label the projection outcomes $\ket{z}$.
So if $U$ describes an experiment implementing the POVM, 
$PU$ also describes the same experiment for some permutation matrix $P$.
Overall, we define the time-energy cost of POVM $\mathcal M$ by
\begin{equation}
\label{eqn-TE-resource-POVM-1}
\maxnorm{\mathcal M}
\equiv
\min_{P,\{V_i\}}
\maxnorm{
P
g(F_1,F_2,\dots,F_K)
}
\end{equation}
where $P$ is some $\KrausDim \ChannelDimO \times \KrausDim \ChannelDimO$ permutation matrix,
and $\maxnorm{Pg}$ is the solution to the partial $U$ problem~\eqref{eqn-problem-partial-with-g}.
As we shall see, the number of zeros padded in $F_i$ (i.e., $m-n$) does not matter.
In the following, we first investigate the special case where $P$ only swaps the POVM elements $\{F_i\}$, i.e., we restrict $P$ to be of the form $\hat{P} \otimes I_m$ where $\hat{P}$ is some $\KrausDim \times \KrausDim$ permutation matrix and $I_m$ is the $m$-dimensional identity matrix.
Then, using the result of this special case, we investigate the case with a general $P$.

\subsection{With arbitrary POVM element labelings}

We first focus on the problem without the optimization over 
$P$ and $\{V_i\}$ (assumed to be fixed),
and with a specific ordering of the POVM elements $(M_k)_{k=1}^K$:
\begin{align}
&
\maxnorm{(M_k)_{k=1}^K}
\nonumber
\\
\equiv&
\maxnorm{
g(F_1,F_2,\dots,F_K)
}
\nonumber
\\
=&
\maxnorm{ g(\hat{Q}F_1Q^\dag,F_2Q^\dag,\dots,F_{\KrausDim}Q^\dag) }
\text{for all } Q
\nonumber
\\
\ge&
\max_{1\leq i \leq n}
\cos^{-1}
\left[
\operatorname{Re}(
(\hat{Q} F_1 Q^\dag) (i,i)
)
\right]
\nonumber
\\
=&
\cos^{-1}
\left[
\min_{1\leq i \leq n}
\operatorname{Re}(
(\hat{Q} F_1 Q^\dag) (i,i)
)
\right]
\label{eqn-max-norm-POVM-lower-bound-no-perm-1}
\end{align}
where 
(i) the third line
is
due to Lemma~\ref{lemma-U-norm-unchanged-by-conjugation-2},
$Q \in U(n)$ and $\hat{Q}$ is of the form in Eq.~\eqref{eqn-U-norm-unchanged-by-conjugation-Qhat};
(ii)
the inequality in the fourth line is due to Eq.~\eqref{eqn-maxnorm-lower-bound-diagonal1}; and 
(iii) the last equality is because
$\cos^{-1}$ is a decreasing function in the range $[0,\pi]$.
Different $Q$ gives different bounds.
With an argument similar to that for
Eq.~\eqref{eqn-maxnorm-lower-bound-diagonal3},
we choose $Q$
to be the right singular matrix of 
$\sqrt{M_1}$
and this gives 
$
\min_i
\operatorname{Re}(
(\hat{Q} F_1 Q^\dag) (i,i)
) \le 
\sigma_\text{min} (F_1)
=
\sigma_\text{min} (\sqrt{M_1})
$ since every element of a unitary matrix (corresponding to the product of $\hat{Q}$, $V_1$, and the left singular matrix of $\sqrt{M_1}$) has a norm no larger than unity,
where $\sigma_\text{min}$ denotes the minimum singular value of its argument.
This shows that
\begin{align}
\maxnorm{(M_k)_{k=1}^K}
&\ge
\cos^{-1} 
\left[
\sigma_\text{min} (\sqrt{M_1}) 
\right].
\label{eqn-max-norm-POVM-lower-bound-no-perm-2}
\end{align}
Since this lower bound is independent of $\{V_i\}$, we have 
\begin{align}
\label{eqn-max-norm-POVM-lower-bound-no-perm-3}
\min_{\{V_i\}}
\maxnorm{
(M_k)_{k=1}^K
}
\ge
\cos^{-1} 
\left[
\sigma_\text{min} (\sqrt{M_1}) 
\right].
\end{align}
On the other hand, this bound can 
be made more stringent by choosing
$V_1$ 
so that 
the product of $\hat{Q}$, $V_1$, and the left singular matrix of $\sqrt{M_1}$ is the identity matrix. 

\medskip

Upper bound ---
We upper bound the above quantity
$
\min_{\{V_i\}}
\maxnorm{(M_k)_{k=1}^K}
$
by letting $V_1$ to be the unitary matrix that transforms the left singular matrix of $\sqrt{M_1}$ to become its right singular matrix.
Applying 
Eq.~\eqref{eqn-maxnorm-special-case2} gives
\begin{align}
\label{eqn-max-norm-POVM-upper-bound-no-perm-1}
\min_{\{V_i\}}
\maxnorm{
g(F_1,F_2,\dots,F_K)
}
\le
\cos^{-1} 
\left[
\sigma_\text{min} (\sqrt{M_1})
\right] .
\end{align}
It is an inequality because we chose one particular $V_1$.
Combining Eqs.~\eqref{eqn-max-norm-POVM-lower-bound-no-perm-3} and \eqref{eqn-max-norm-POVM-upper-bound-no-perm-1} gives
\begin{align}
\label{eqn-max-norm-POVM-exact-no-perm-1}
\min_{\{V_i\}}
\maxnorm{
g(F_1,F_2,\dots,F_K)
}
=
\cos^{-1} 
\left[
\sigma_\text{min} (\sqrt{M_1})
\right] .
\end{align}

We now consider the minimization over permutations.
For the special case that $P$ permutes only the POVM elements, we have the following.
\begin{theorem}
{\rm
\begin{align}
&\maxnorm{\{M_{k}\}_{k=1}^K} \nonumber
\\
\equiv&
\min_{{\pmb\pi}}
\min_{\{V_i\}}
\maxnorm{
g(F_{{\pmb\pi}(1)},F_{{\pmb\pi}(2)},\dots,F_{{\pmb\pi}(K)})
}
\nonumber
\\
=&
\min_{1\leq k \leq K}
\cos^{-1} 
\left[
\sigma_\text{min} (\sqrt{M_k}) 
\right] ,
\label{eqn-TE-POVM-permute-element-index}
\end{align}
where $\pmb\pi$ denotes the ordering function.
}
\end{theorem}

\subsection{With arbitrary detection event labelings}

We now consider general permutations over all detection events of all POVM elements and bound 
$\maxnorm{\mathcal M}$
in Eq.~\eqref{eqn-TE-resource-POVM-1}.
Essentially, 
the permutation $P$ in $Pg(F_1,F_2,\dots,F_K)$ serves to produce a new top-left $n \times n$ block which we denote as $\tilde{F}$.
We may reuse Eqs.~\eqref{eqn-max-norm-POVM-lower-bound-no-perm-1} and \eqref{eqn-max-norm-POVM-lower-bound-no-perm-2} 
with this $\tilde{F}$ in place of $F_1$.
Depending on how we choose $Q$ in Eq.~\eqref{eqn-max-norm-POVM-lower-bound-no-perm-1}, we have two methods to lower bound
$\maxnorm{\mathcal M}$.
In general, we may take the maximum of two bounds of the two methods [cf. Eqs.~\eqref{eqn-general-permutation-method-Q-identity-1}, \eqref{eqn-column-Mj-bound}, \eqref{eqn-TE-POVM-general-lower-bound}, and \eqref{eqn-TE-POVM-general-lower-bound2}].

Later, we will apply Method~1 in the examples in Sec.~\ref{sec-example-bell-state} and Method~2 in the examples in Secs.~\ref{sec-example-linear-optics} and \ref{sec-example-usd}.

\subsubsection{Method 1}
Let us consider the first way to bound
$\maxnorm{\mathcal M}$
in Eq.~\eqref{eqn-TE-resource-POVM-1}.
Starting from
Eq.~\eqref{eqn-max-norm-POVM-lower-bound-no-perm-1} with $Q$ being the identity matrix, we have
\begin{align*}
\maxnorm{\mathcal M}
&=
\min_{P,\{V_j\}}
\maxnorm{
P
g(F_1,F_2,\dots,F_K)
}
\\
&\ge
\cos^{-1}
\Big[
\max_{P,\{V_j\}}
\min_{1\leq i \leq n}
\operatorname{Re}(
\tilde{F}  (i,i)
)
\Big]
\equiv \cos^{-1} A
\end{align*}
where 
we used the fact that
$\cos^{-1}$ is a decreasing function in the range $[0,\pi]$.
Using the max-min inequality (see, e.g., Ref.~\cite{Boyd:2004}),
\begin{align*}
A
&\le
\min_{1\leq i \leq n}
\max_{P,\{V_j\}}
\operatorname{Re}(
\tilde{F}  (i,i)
)
\\
&=
\min_{1\leq i \leq n}
\max_{j}
\lVert \sqrt{M_j}(*,i) \rVert_2
\end{align*}
where the term on the RHS of the second line is the $\ell_2$-norm of the $i$th column of $\sqrt{M_j}$.
The second line is due that 
whenever we choose through $P$ the $i$th row of $\tilde{F}$ to be the $l$th row of the $j$th POVM element $F_j=V_j \sqrt{M_j}$, we can always maximize this $l$th row's $i$th column element by choosing the best rotation $V_j$.
The best rotation concentrates all elements of the $i$th column of $\sqrt{M_j}$ to the $l$th row.
This gives one way to lower bound $\maxnorm{\mathcal M}$:
\begin{theorem}
{\rm
\begin{align}
\maxnorm{\mathcal M}
\ge
\cos^{-1}
\left[
\min_{1\leq i \leq n}
\max_{1\le j \le K}
\lVert \sqrt{M_j}(*,i) \rVert_2 .
\right] .
\label{eqn-general-permutation-method-Q-identity-1}
\end{align}
}
\end{theorem}
This lower bound is easy to compute, by first obtaining the norm of every column of all $\sqrt{M_j}$ and then comparing them.
\begin{corollary}
\label{cor-column-Mj-bound}
{\rm
If there is a $\sqrt{M_j}$ having a column with norm $c \ge 1/\sqrt{2}$,
\begin{align}
\maxnorm{\mathcal M}
\ge
\cos^{-1} (c) .
\label{eqn-column-Mj-bound}
\end{align}
}
\end{corollary}
\begin{proof}
For any POVM, the trace-preserving constraint implies that
$\sum_{j=1}^K \lVert \sqrt{M_j}(*,i) \rVert_2^2 =1$.
Thus, 
$\max_{j} \lVert \sqrt{M_j}(*,i) \rVert_2 =c$.
Finally, we can neglect the minimization over $i$ since every $i$ serves as a lower bound.
\end{proof}

\subsubsection{Method 2}

Let us consider the second way to bound
$\maxnorm{\mathcal M}$
in Eq.~\eqref{eqn-TE-resource-POVM-1}.
We start from
Eqs.~\eqref{eqn-max-norm-POVM-lower-bound-no-perm-1} and \eqref{eqn-max-norm-POVM-lower-bound-no-perm-2} 
with $\tilde{F}$ in place of $F_1$.
Note that the upper bound in Eq.~\eqref{eqn-max-norm-POVM-upper-bound-no-perm-1} does not apply here since we now do not have the unitary degree of freedom on the left (i.e., $V_1$) to make the top-left $n \times n$ block of $Pg(F_1,F_2,\dots,F_K)$ Hermitian.
The $n$ rows of $\tilde{F}$ are constructed by selecting rows coming from any Kraus operators ${F_i}$ of Eq.~\eqref{eqn-general-form-Kraus-operator}, $i=1,\dots,K$ (not necessarily from the same element).
Thus we have the following.
\begin{theorem}
\label{thm-maxnorm-lower-bound-all-permutations}
{\rm
\begin{align}
\maxnorm{\mathcal M}
\ge
\min_{P,\{V_i\}}
\cos^{-1} 
\left[
\sigma_\text{min}(\tilde{F}) 
\right]
\label{eqn-TE-POVM-general-lower-bound}
\end{align}
where $P$ denotes the selection of the rows of $\tilde{F}$ coming from any 
Kraus operators
${F_i}$ of Eq.~\eqref{eqn-general-form-Kraus-operator}, $i=1,\dots,K$.
}
\end{theorem}
In general we need to iterate over all permutations of the rows to find the best $\tilde{F}$ to achieve the minimum on the RHS.
Also, this lower bound may not be tight.
On the other hand, we may bound $\sigma_\text{min}(\tilde{F})$ as follows.
First, it is no larger than the norm of any row $j$ of $\tilde{F}$:
\begin{align}
\tilde{F}(j,*) \tilde{F}(j,*)^\dag&=
[W_L(j,*) S W_R^\dag]
[W_R S^\dag W_L(j,*)^\dag]
\nonumber
\\
&=
\sum_{i=1}^n |W_L(j,i)|^2 \sigma_i^2(\tilde{F})
\nonumber
\\
&\ge \sigma_\text{min}^2(\tilde{F})  \hspace{.7cm} \text{for }\: 1 \le j \le n
\label{eqn-bound-of-sigma-F}
\end{align}
where we take the singular value decomposition $\tilde{F}=W_L S
W_R^\dag$
and $\sigma_i(\tilde{F}), i=1,\dots,n$ are the diagonal elements of $S$.
Second, $\sigma_\text{min}(\tilde{F})$ is no larger than the minimum singular value of 
any subset of rows of $\tilde{F}$.
This follows by simply multiplying the left singular matrix of this submatrix to the left of $\tilde{F}$ and applying the above result to this new $\tilde{F}$~\footnote
{For example, suppose that the subset of rows comes from the first two rows of $\tilde{F}$ and $R$ is the $2\times 2$ left singular matrix of it.
Then, let $\tilde{F}'=\begin{bmatrix}R^\dag&0\\0&{\mathbf 1}\end{bmatrix} \tilde{F}$ and apply Eq.~\eqref{eqn-bound-of-sigma-F} to $\tilde{F}'$.
Note that $\tilde{F}'$ and $\tilde{F}$ have the same singular values.}.
Thus, 
we construct $\tilde{F}$ by taking rows from $\{F_i\}$ with as large singular values as possible which 
can be done by choosing $V_i$ to cancel out the left singular matrix of $\sqrt{M_i}$.
Therefore, a strategy to find a lower bound of 
$\maxnorm{\mathcal M}$
in Eq.~\eqref{eqn-TE-resource-POVM-1} is 
the following.
\begin{lemma}
{\rm
Order all singular values of all $\sqrt{M_i}$, $i=1,\dots,K$, and obtain the $n$th largest singular value $\sigma_n$.
Then,
\begin{align}
\maxnorm{\mathcal M}
\ge
\cos^{-1} (\sigma_n) .
\label{eqn-TE-POVM-general-lower-bound2}
\end{align}
}
\end{lemma}
We remark that we do not take into account the amounts of overlaps between the rows of $\tilde{F}$ when we select them and thus this lower bound can be loose in some cases [i.e., the RHS of Eq.~\eqref{eqn-TE-POVM-general-lower-bound2} is lower than that of Eq.\eqref{eqn-TE-POVM-general-lower-bound}].
As an extreme example, two rows of $\tilde{F}$ come from different $\sqrt{M_i}$ and $\sqrt{M_j}$ such that the one row is a scalar multiple of each other.
This makes the smallest singular value of $\tilde{F}$ zero instead of $\sigma_n$.
In general, we need to go through all permutations in Eq.~\eqref{eqn-TE-POVM-general-lower-bound} to obtain good lower bounds.

We consider optimality for special cases.
\begin{lemma}
{\rm
If an $\tilde{F}$ can be found such that 
the RHS of Eq.~\eqref{eqn-TE-POVM-general-lower-bound2} is equal to that of Eq.~\eqref{eqn-TE-POVM-general-lower-bound} [i.e., $\sigma_\text{min}(\tilde{F})=\sigma_n$], such an $\tilde{F}$ is the minimizing $\tilde{F}$ for Eq.~\eqref{eqn-TE-POVM-general-lower-bound}.
}
\end{lemma}
Furthermore, if the minimizing $\tilde{F}$ in Eq.~\eqref{eqn-TE-POVM-general-lower-bound} is Hermitian, 
we upper bound $\maxnorm{\mathcal M}$
in Eq.~\eqref{eqn-TE-resource-POVM-1} 
by using Eq.~\eqref{eqn-maxnorm-special-case2} [similar to the argument for Eq.~\eqref{eqn-max-norm-POVM-upper-bound-no-perm-1}]:
$$
\maxnorm{\mathcal M}
\le
\cos^{-1} 
\left[
\sigma_\text{min}(\tilde{F})
\right] .
$$
Combining this with Eq.~\eqref{eqn-TE-POVM-general-lower-bound} gives the following.
\begin{theorem}
\label{thm-TE-POVM-general-special-case}
{\rm
If the minimizing $\tilde{F}$ for Eq.~\eqref{eqn-TE-POVM-general-lower-bound} is Hermitian,
\begin{align}
\label{eqn-TE-POVM-general-special-case1}
\maxnorm{\mathcal M}
=
\cos^{-1}
\left[
\sigma_\text{min}(\tilde{F}) 
\right] .
\end{align}
}
\end{theorem}

\section{Examples}
\label{sec-examples}

We compute the time-energy costs for a few quantum measurements and also compare them with the costs of some actual experiments based on the linear optical components used.
We do not consider the detectors in all time-energy cost calculations below.

\subsection{Time-energy cost for $\myUgrp(2)$}
The most general unitary operator in $\myUgrp(2)$ can be implemented by a beam splitter (BS) 
with the freedom to choose the reflectivity and phase as follows~\cite{Zeilinger:1994:two-particle}:
\begin{align}
U_\text{BS}=
\exp (i \chi)
\begin{bmatrix}
r & i t^*\\
i t & r^* 
\end{bmatrix}
\label{eqn-U-BS}
\end{align}
where $\chi$ is an arbitrary real number, and $r$ and $t$ are the reflection and transmission amplitudes (complex) with $|r|^2+|t|^2=1$.
We seek the most efficient $U_\text{BS}$ for a fixed reflectivity $|r|$ based on 
$\maxnorm{U_\text{BS}}$.
The eigenvalues of $U_\text{BS}$ are
$\exp(i \chi) \left[ \text{Re}(r) \pm i \sqrt{ |t|^2+\text{Re}^2(r)} \right]$.
It can be easily seen that the best parameters are $\chi=0$ and $r=|r|$, giving
\begin{equation}
\label{eqn-maxnorm-U-BS}
\maxnorm{U_\text{BS}}=\cos^{-1} |r|.
\end{equation}

\begin{figure}[t!]
\begin{center}
\includegraphics[width=.5\columnwidth]{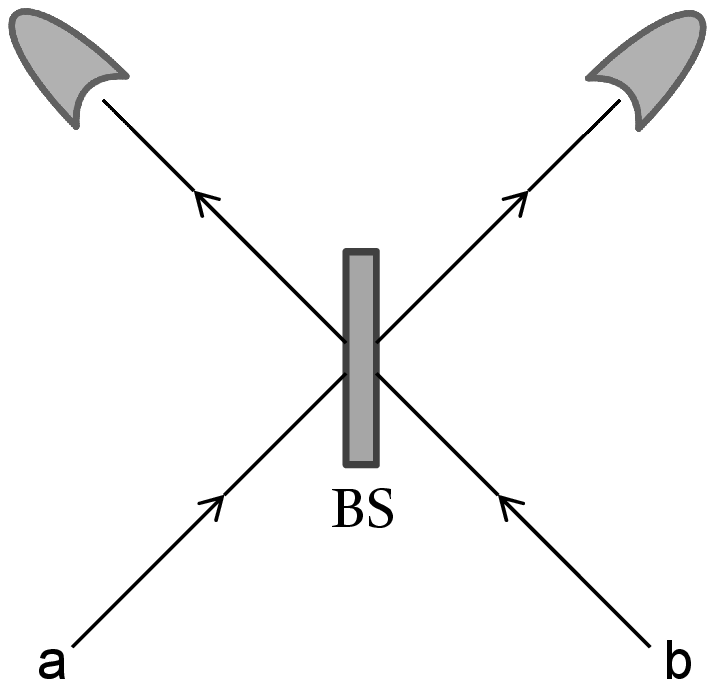}
\caption{
\label{fig-Bell1}
Bell measurement for $\ket{\Psi^-}$.
}
\end{center}
\end{figure}

\subsection{Time-energy cost for Bell state analysis}
\label{sec-example-bell-state}
\subsubsection{One Bell state}
A 50-50 beam splitter can be used to project the two-photon input state onto the singlet Bell state~\cite{Braunstein:1995:Bell,Lutkenhaus:1999:Bell} (see Fig.~\ref{fig-Bell1}).
The four Bell states are
\begin{align*}
\ket{\Psi^\pm} &= ( \ket{\updownarrow}_a \ket{\leftrightarrow}_b \pm \ket{\leftrightarrow}_a \ket{\updownarrow}_b )  /\sqrt{2}
\\
\ket{\Phi^\pm} &= ( \ket{\updownarrow}_a \ket{\updownarrow}_b \pm \ket{\leftrightarrow}_a \ket{\leftrightarrow}_b )  /\sqrt{2}
\end{align*}
where two photons are in modes $a$ and $b$, and $\ket{\updownarrow}$ and $\ket{\leftrightarrow}$ are single-photon states with vertical and horizontal polarizations.
Two detectors are installed at the two output ports of the BS, and when both report a click, the input state is collapsed to the singlet state $\ket{\Psi^-}$.
This simple setup cannot make projections onto the other three Bell states which is possible with more complicated setups~\cite{Braunstein:1995:Bell,Lutkenhaus:1999:Bell}.
Based on the previous analysis resulting in Eq.~\eqref{eqn-maxnorm-U-BS}, the time-energy cost to collapse a two-photon state to $\ket{\Psi^-}$ with this simple setup is $\cos^{-1} (1/\sqrt{2})=\pi/4$ using the fact that it is a 50-50 BS.

Let us consider the time-energy cost for the ideal measurement with a projection onto $\ket{\Psi^-}$.
Obviously, there is a POVM element $\ket{\Psi^-}\bra{\Psi^-}$ and following
Corollary~\ref{cor-column-Mj-bound}, we can see that a column of it has norm $1/\sqrt{2}$.
So, by Eq.~\eqref{eqn-column-Mj-bound}, the cost lower bound is $\pi/4$.
Therefore, the above implementation with one BS is optimal since it achieves this bound.

\begin{figure}[t!]
\begin{center}
\includegraphics[width=.75\columnwidth]{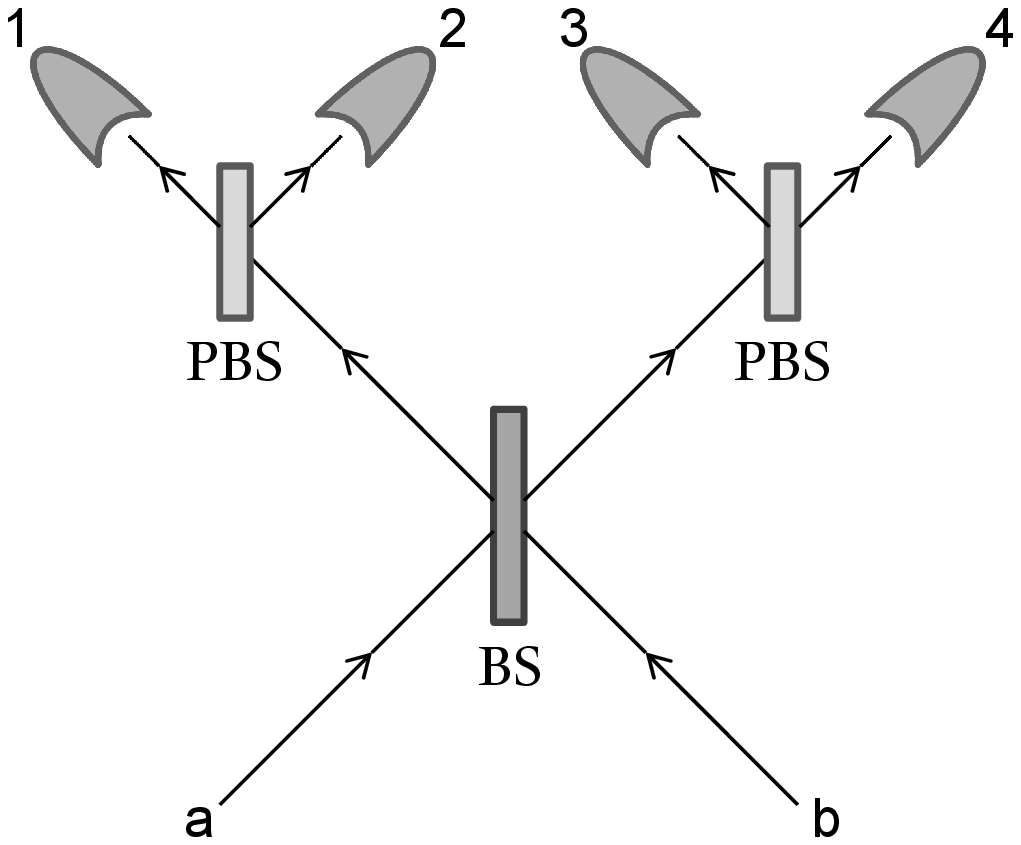}
\caption{
\label{fig-Bell2}
Bell measurement for $\ket{\Psi^-}$ and $\ket{\Psi^+}$.
}
\end{center}
\end{figure}

\subsubsection{Two Bell states}
A more complicated setup, the Innsbruck detection scheme ~\cite{Weinfurter:1994:Bell,Braunstein:1995:Bell,Michler:1996:Bell} , as shown in Fig.~\ref{fig-Bell2}, can project onto two Bell states.
Coincidence detections at detectors 1 and 4 or at 2 and 3 correspond to projection onto $\ket{\Psi^-}$.
Coincidence detections at detectors 1 and 2 or at 3 and 4 correspond to projection onto $\ket{\Psi^+}$.
The event of having two particles at any one of the four detectors could have been triggered by $\ket{\Phi^+}$ or $\ket{\Phi^-}$.

The time-energy cost for the ideal measurement with projections onto $\ket{\Psi^\pm}$ is lower bounded by $\pi/4$, argued as above.
We construct a $U$ with these two projections in order to obtain an upper bound:
\begin{align*}
U&=\ket{0} \bra{\Psi^-}_{ab} + \ket{1} \bra{\Psi^+}_{ab} + \ket{2} \bra{ \updownarrow \updownarrow}_{ab} +  \ket{3} \bra{ \leftrightarrow \leftrightarrow}_{ab}
\\
&=
\begin{bmatrix}
\frac{1}{\sqrt{2}} & -\frac{1}{\sqrt{2}} & 0 & 0
\\
\frac{1}{\sqrt{2}} & \frac{1}{\sqrt{2}} & 0 & 0
\\
0 & 0 & 1 & 0
\\
0 & 0 & 0 & 1
\end{bmatrix}
\end{align*}
where $U$ 
acts on states specified in the basis
$\{\ket{\updownarrow \leftrightarrow},
\ket{\leftrightarrow \updownarrow }, \ket{ \updownarrow \updownarrow}, \ket{ \leftrightarrow \leftrightarrow}\}$
and produces the detection events labeled as $\ket{j}$, $j=1,2,3,4$.
It is clear that $\maxnorm{U}=\pi/4$.
Therefore, the time-energy cost for the ideal measurement with projections onto $\ket{\Psi^\pm}$ is $\pi/4$.

Comparison between the time-energy cost for the ideal measurement and the cost for the actual implementation may subject to interpretations.
We may compute the overall cost for all the linear optics devices responsible for (i) only the transformation or (ii) the transformation and detection.
The detection part is for detecting the horizontal and vertical qubit states and it consists of a polarizing beam splitter (PBS) and two detectors.
One may argue that this part is used anyway to detect the original input qubit when no transformation is involved and so it should not be included.
On the other hand, 
including the detection part in the overall cost also makes sense 
since sometimes it is not needed (for example in the one Bell state measurement); 
also, it is specific to linear optics implementations and 
we may want to include all costs due to this type of implementations
when our consideration is not restricted to this type.
Here, we adopt interpretation (ii) since it is the presence of the two PBS that 
enables the projections onto two Bell states.
As such, the time-energy cost for the Innsbruck scheme certainly costs more than $\pi/4$ since it contains a 50-50 BS and two PBS,
and the BS already costs $\pi/4$.
To find time-energy cost for a PBS, consider its unitary representation
for transforming the polarization states of the two input modes:
\begin{align}
U_\text{PBS}=
e^{i \chi}
\begin{bmatrix}
\ket{\updownarrow}\bra{\updownarrow} & \ket{\leftrightarrow}\bra{\leftrightarrow} 
\\
\ket{\leftrightarrow}\bra{\leftrightarrow} & \ket{\updownarrow}\bra{\updownarrow}
\end{bmatrix}
\label{eqn-U-PBS}
\end{align}
which has eigenvalues $-e^{i \chi}$, $e^{i \chi}$, $e^{i \chi}$, and $e^{i \chi}$.
With $\chi=\pi/2$, the smallest time-energy cost is $\maxnorm{U_\text{PBS}}=\pi/2$.

\subsection{Time-energy cost for general measurements on linear optical qubits}
\label{sec-example-linear-optics}

\begin{figure}[t!]
\begin{center}
\includegraphics[width=\columnwidth]{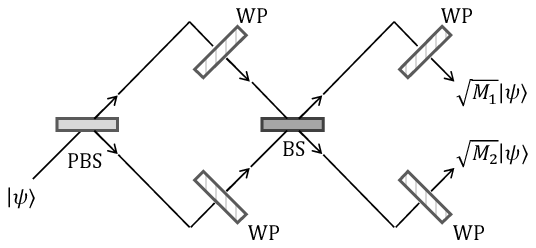}
\caption{
\label{fig-rank-2}
An implementation of the POVM in Eq.~\eqref{eqn-example-linear-optics-POVM}.
}
\end{center}
\end{figure}

A scheme for general measurements on linear optical qubits was proposed in Ref.~\cite{Ota:2012:optics-measurement}.
We analyze the time-energy cost for their measurement implementation 
shown in Fig.~\ref{fig-rank-2} (which is Fig.~1 of Ref.~\cite{Ota:2012:optics-measurement}), consisting of, sequentially, 
a PBS, two wave plates (WP), a BS, and two WP.
The input state is polarization encoded: 
$\ket{\psi}=c_\text{H} \ket{\leftrightarrow}+c_\text{V} \ket{\updownarrow}$.
The POVM elements to be implemented are
\begin{equation}
\begin{aligned}
M_1&=
\cos^2 \varphi \ket{m_+}\bra{m_+}+\sin^2 \varphi \ket{m_-}\bra{m_-}
\\
M_2&=
\sin^2 \varphi \ket{m_+}\bra{m_+}+\cos^2 \varphi \ket{m_-}\bra{m_-}
\end{aligned}
\label{eqn-example-linear-optics-POVM}
\end{equation}
where we assume $w=0$ in the implementation of Ref.~\cite{Ota:2012:optics-measurement}.
Here, $\{\ket{m_\pm}\}$ form an orthonormal basis.
We assume that $0\le\varphi\le\pi/2$.

We first compute the time-energy cost for the implementation.
For simplicity, we only consider the PBS and BS, which will give us a cost lower bound.
In the implementation, the PBS is the one in Eq.~\eqref{eqn-U-PBS} and the BS is 
the one in Eq.~\eqref{eqn-U-BS} with reflectivity $|r|=\cos \varphi$.
Thus,
$\maxnorm{U_\text{PBS}}=\pi/2$ and
$\maxnorm{U_\text{BS}}=\varphi$.

The total evolution time $t_\text{tol}$ is split between the PBS and BS:
\begin{equation}
\label{eqn-example-optics-total-time}
t_\text{tol}=t_\text{PBS}+t_\text{BS}
\end{equation}
and the total energy is thus
$\pi/2/t_\text{PBS}+\varphi/t_\text{BS}$.
The optimal split between $t_\text{PBS}$ and $t_\text{BS}$ is found by 
\dmathX2{
E_\text{tol}^\text{impl} &=& 
\min_{t_\text{PBS},t_\text{BS}} &
\frac{\pi}{2 t_\text{PBS}}+\frac{\varphi}{t_\text{BS}}
&eqn-example-linear-optics-Eimpl\cr
&&
\text{s.t.} &
t_\text{tol}=t_\text{PBS}+t_\text{BS}
\cr
}
which can be solved analytically easily.

Next, we obtain
the time-energy cost for the POVM in Eq.~\eqref{eqn-example-linear-optics-POVM} 
by solving
Eq.~\eqref{eqn-TE-POVM-general-lower-bound}.
We can solve it by going through all $12$ permutations for $\tilde{F}$ to get
$$
\maxnorm{\mathcal M}
\ge
\begin{cases}
\varphi & \text{if } 0 \le \varphi < \pi/4 \\
\frac{\pi}{2}-\varphi & \text{if } \pi/4 \le \varphi < \pi/2
\end{cases}
$$
with, for the case $\varphi \in [0,\pi/4)$,
$$
\tilde{F}=
\begin{bmatrix}
\cos \varphi & 0\\
0 & \cos \varphi
\end{bmatrix}
$$
which is formed by taking the first row of $\sqrt{M_1}$ and the second row of $\sqrt{M_2}$, and 
for the case $\varphi \in [\pi/4,\pi/2)$,
$$
\tilde{F}=
\begin{bmatrix}
\sin \varphi & 0\\
0 & \sin \varphi
\end{bmatrix}
$$
which is formed by taking 
the first row of $\sqrt{M_2}$ and the second row of $\sqrt{M_1}$.
Since these minimizing $\tilde{F}$ are diagonal for both cases,
Theorem~\ref{thm-TE-POVM-general-special-case} implies that
\begin{align*}
\maxnorm{\mathcal M}
&=
\begin{cases}
\varphi & \text{if } 0 \le \varphi < \pi/4 \\
\frac{\pi}{2}-\varphi & \text{if } \pi/4 \le \varphi < \pi/2
\end{cases}.
\end{align*}
Using $\maxnorm{\mathcal M}$ as the time-energy cost, we have
\begin{align}
E_\text{tot}^\text{ideal}
\equiv
\frac{
\maxnorm{\mathcal M}
}
{t_\text{tot}}.
\label{eqn-example-linear-optics-Eideal}
\end{align}

We can see how much more energy is used for the same time $t_\text{tot}$ in the actual implementation compared to the ideal one by computing $E_\text{tot}^\text{impl}/E_\text{tot}^\text{ideal} \triangleq r(\varphi)$ which turns out to be independent of $t_\text{tot}$.
Figure~\ref{fig-linear-optics} shows that result and it can be seen that the PBS causes a significant increase in the energy cost for small and large $\varphi$.

\begin{figure}[t!]
\begin{center}
\includegraphics[width=\columnwidth]{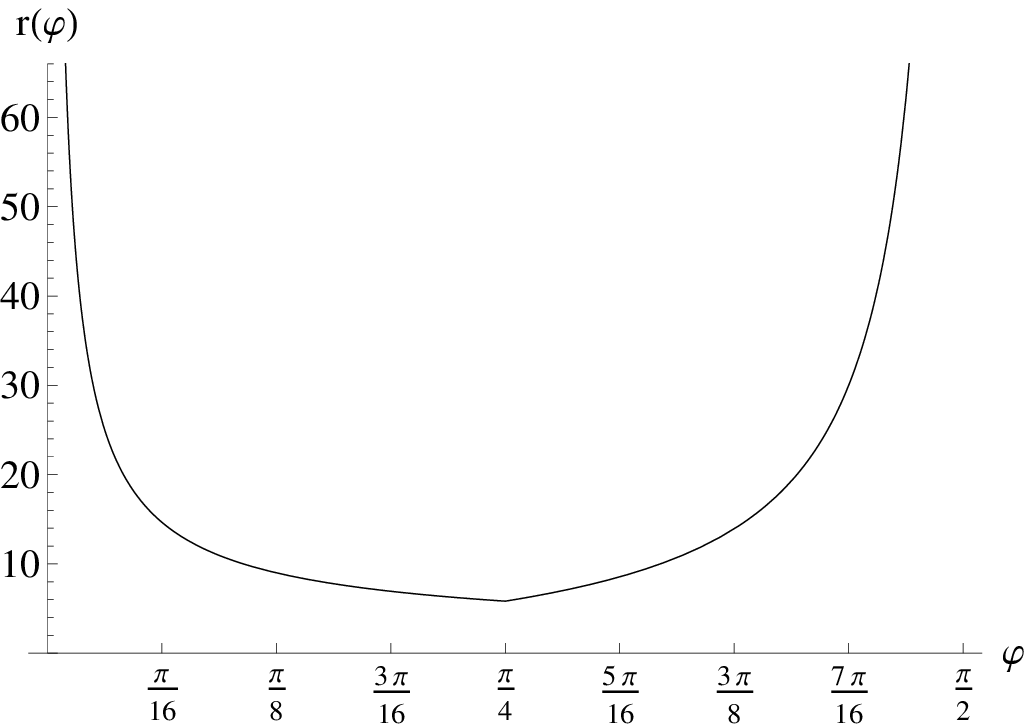}
\caption{
\label{fig-linear-optics}
Ratio of the energy of the measurement implementation with linear optics (Eq.~\eqref{eqn-example-linear-optics-Eimpl}) to the minimum energy (Eq.~\eqref{eqn-example-linear-optics-Eideal}):
$r(\varphi)=E_\text{tot}^\text{impl}/E_\text{tot}^\text{ideal}$.
}
\end{center}
\end{figure}

\subsection{Time-energy cost for unambiguous state discrimination}
\label{sec-example-usd}

We analyze the time-energy cost for unambiguous state discrimination (USD) of geometrically uniform (GU) states~\cite{Eldar:2003:USD}.
A set of GU states generated by a single normalized state $\ket{\phi} \in \IC^n$ is $\mathcal S=\{\ket{\phi_i}=U_i \ket{\phi}, U_i \in \mathcal G \}$, where $\mathcal G$ is a finite group of unitary matrices $\{U_i \in \myUgrp(n), i=1,\dots,K-1\}$ such that $U_i U_j \in \mathcal G$ and $U_i^\dag \in \mathcal G$ for all $i,j$.
We assume the states in $\mathcal S$ have equal prior probability $1/(K-1)$.
Theorem~4 of Ref.~\cite{Eldar:2003:USD} proves that the POVM $\mathcal M$ that unambiguously discriminates these states 
with the minimum inconclusive result consists of $K$ POVM elements
\begin{align*}
M_i&=p \kett{\tilde{\phi}_i} \braa{\tilde{\phi}_i}, \text{    for } i=1,\dots,K-1
\\
M_K&=I-\sum_{i=1}^{K-1} M_i
\end{align*}
where $\{ \kett{\tilde{\phi}_i} = U_i \kett{\tilde{\phi}} , U_i \in \mathcal G\}$,
$\kett{\tilde{\phi}}=(\Phi \Phi^\dag)^{-1} \ket{\phi}$,
$\Phi$ is a matrix of columns $\ket{\phi_i}$, and 
$\sqrt{p}$ is the smallest singular value of $\Phi$.
Here, $(\Phi \Phi^\dag)^{-1}$ is the Moore-Penrose pseudoinverse of $\Phi \Phi^\dag$.
Note that $\kett{\tilde{\phi}_i}$ is not necessarily normalized.
It turns out that this optimal USD measurement produces equal probabilities for detecting each state in $\mathcal S$.
This detection probability is
\begin{align*}
&
\operatorname{Pr}(
\text{concluding } i
\text{ given that }\ket{\phi_i} \text{ is emitted})
\\
=&
\bra{\phi_i} M_i\ket{\phi_i}
=
p
=
\sigma_\text{min}^{2}(\Phi).
\end{align*}

We are interested in the time-energy cost of this USD measurement.
We apply Eq.~\eqref{eqn-TE-POVM-general-lower-bound2} to lower bound $\maxnorm{\mathcal M}$.
The single non-zero singular value of $M_i, i=1,\dots,K-1$ is
\begin{align*}
p \braakett{\tilde{\phi}_i}{\tilde{\phi}_i}
&=
p \:
\bra{\phi}
(\Phi \Phi^\dag)^{-2} \ket{\phi}
\\
&=
\sigma_\text{min}^{2}(\Phi)
\bra{\phi}
(\Phi \Phi^\dag)^{-2} \ket{\phi}.
\end{align*}
Now, let's focus on $M_K$.
Note that $T=\sum_{i=1}^{K-1} M_i$ has rank at most $K-1$ and
thus $M_K$ has at least $n-K+1$ eigenvalues of one.
Also, 
$T$ has an eigenvalue of one since otherwise we would have increased $p$ and the original POVM was not optimal.
This means that $M_K$ has 
at least one eigenvalue of zero.
We need to find the $n$th largest singular value $\sigma_n$ among all singular values of all $\sqrt{M_i}$.
The first $n-K+1$ largest singular values are equal to one coming from $\sqrt{M_K}$.
The next $K-2$ singular values come from any of $\sqrt{M_i}, i=1,\dots,K$.
And the next one (i.e., the $n$th one)
must be 
$$
\sigma_n=
\sigma_\text{min}(\Phi)
\sqrt{
\bra{\phi}
(\Phi \Phi^\dag)^{-2} \ket{\phi}
}.
$$
coming from any one of $\sqrt{M_i}, i=1,\dots,K-1$.
Therefore, the time-energy cost for the optimal USD measurement $\mathcal M$ for GU states with equal prior probabilities is
\begin{align}
\maxnorm{\mathcal M}
\ge
\cos^{-1} 
\left[
\sigma_\text{min}(\Phi)
\sqrt{
\bra{\phi}
(\Phi \Phi^\dag)^{-2} \ket{\phi}
}
\right].
\label{eqn-cost-USD-lower-bound}
\end{align}

As a numerical example, we consider $\bar{K}\equiv K-1$ coherent states of the same mean photon number $|\alpha|^2$ but with different phases:
$$
\ket{\phi_j}=e^{|\alpha|^2/2} \sum_{m=0}^\infty \frac{\alpha_j^m}{\sqrt{m!}} \ket{m}
$$
where $j=1,\dots,\bar{K}$, $\alpha_j=\alpha \: e^{i 2 \pi  (j-1)/\bar{K}}$, and $\ket{m}$ are the boson number states.
Note that $\ket{\phi_j}=U^j \ket{\phi_0}$ with
$$
U=\sum_{m=0}^\infty e^{i 2\pi m/\bar{K}} \ket{m}\bra{m}.
$$
Therefore, $\ket{\phi_j}$ are GU states.
We compute the lower bound of the time-energy cost for the optimal USD measurement $\mathcal M$ that distinguishes $\ket{\phi_j}$, $j=1,\dots,\bar{K}$.
For simplicity, we approximate $\ket{\phi_j}$ and $U$ by truncating the sums to the first $50$ terms, which is reasonable since we consider $|\alpha|$ to be small.
Thus, we consider the states to be 50-dimensional.
The lower bounds of the time-energy costs using Eq.~\eqref{eqn-cost-USD-lower-bound} is shown in Fig.~\ref{fig-usd}.
Among the four intensities plotted, the USD measurement corresponding to the highest intensity case has the smallest lower bound of the time-energy cost and thus may actually require a smaller time-energy cost.
Also, the figure suggests that it takes more time-energy cost to distinguish a higher number of states.
Interestingly, the cost lower bound saturates to some value as the number of states increases.
This may indicate that a finite time-energy resource is enough to distinguish any number of states (for a fixed mean photon number).
\begin{figure}[t!]
\begin{center}
\includegraphics[width=\columnwidth]{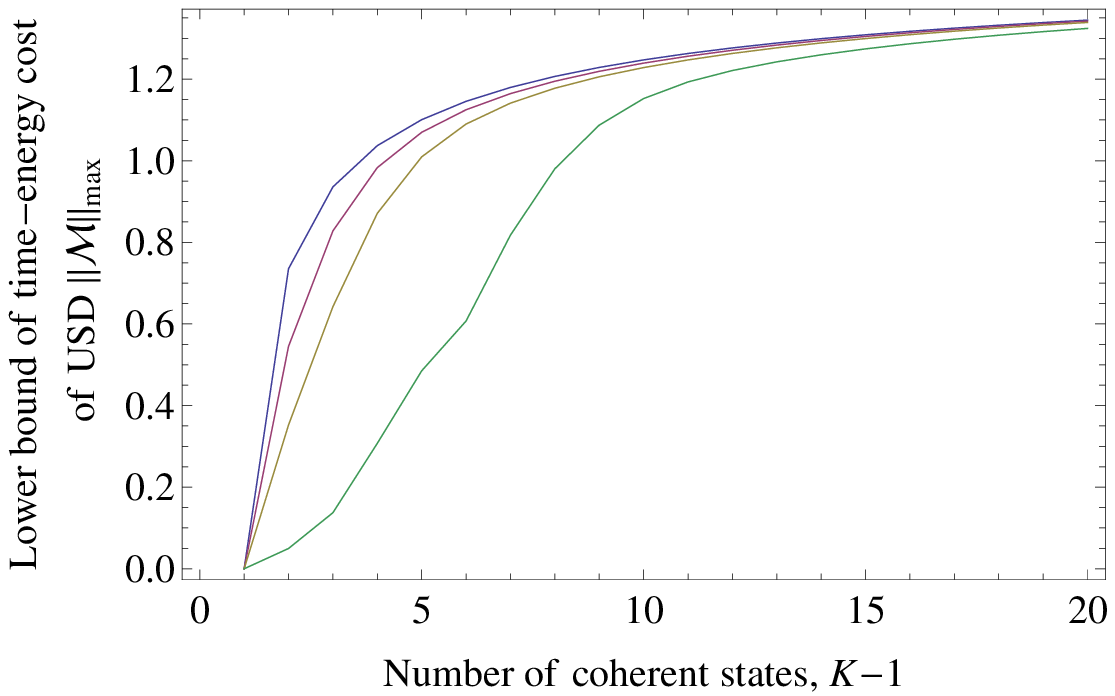}
\caption{
\label{fig-usd}
Lower bound of the time-energy cost for the optimal USD for distinguishing $K-1$ symmetric coherent states (Eq.~\eqref{eqn-cost-USD-lower-bound}).
The four curves from top to bottom correspond to mean photon number $|\alpha|^2=.1, .5, 1, 3$.
}
\end{center}
\end{figure}

\section{Concluding remarks}
\label{sec-conclusion}

We propose and investigate the time-energy cost for POVMs, along a similar line as our previous work for unitary transformations and quantum channels.
We motivate our definition for the time-energy cost by a TEUR.
To find the cost, a POVM is regarded as a quantum channel 
embedded in a unitary transformation in a larger Hilbert space.
The minimum cost among all unitary transformations implementing this POVM is the cost of the POVM.
We proved formulae for computing POVM time-energy cost based on the POVM elements.
When we only optimize over the ordering of the POVM elements in the larger unitary transformation, we obtain the cost in Eq.~\eqref{eqn-TE-POVM-permute-element-index} which depends on the minimal singular value of some element.
A POVM element may correspond to multiple detection events.
When we also optimize over the detection events of the POVM elements, we obtain lower bounds to the cost in Eq.~\eqref{eqn-general-permutation-method-Q-identity-1} and \eqref{eqn-TE-POVM-general-lower-bound}.
Under a special case satisfying the Hermitian condition, the cost is given by Eq.~\eqref{eqn-TE-POVM-general-special-case1}.

The time-energy cost of a POVM can be used as a benchmark for the efficiency of actual experiments.
We compared the costs of the ideal POVMs and the actual linear optics experiments for the Bell measurements and a POVM with rank-2 elements.
We saw that the Bell measurement for one Bell state is optimal but that for two Bell states is not.
Also, the implementation for the POVM with rank-2 elements may not be optimal.
We computed the lower bound to the time-energy cost for the optimal USD for distinguishing symmetric coherent states.
Our result suggests that more time-energy resource is needed to distinguish more states, in line with intuition, but interestingly the cost lower bound saturates as the number of states increases.
This may indicate that a finite time-energy resource is enough to distinguish any number of states.

\section*{Acknowledgments}%
We thank H.-K. Lo for enlightening discussion.
This work is supported in part by
RGC under Grant No. 700712P from the HKSAR Government.

\appendix

\section{Summary of previous work}
\label{app-summary}

We summarize the results of Ref.~\cite{Fung:2013:Time-energy} for quantum channels that are useful to this work.
Given a matrix $U$, the submatrix formed from columns $a$ to $b$ inclusively is denoted by $U_{[a,b]}$ with $a \le b$.

\subsection{Partial $U$ problem with $n$ vectors}

Solving problem~\eqref{eqn-problem-partial-with-g} means
finding $U \in \myUgrp(r)$ where $r=\KrausDim m$ with
the smallest $\maxnorm{U}$ of the form
\begin{align}
\label{eqn-U-with-missing-columns}
U=
\begin{bmatrix}
\ket{b_1} & \ket{b_2} & \dots & \ket{b_n} 
& * & * & \cdots & *
\end{bmatrix}
\end{align}
where the first $n$ columns
are orthogonal and $n \le r$.
We formulated this problem in Ref.~\cite{Fung:2013:Time-energy} as 
finding such $U$ that transforms $\ket{e_i} \longrightarrow \ket{b_i}$  for all $i=1,\ldots,n$:
\dmathX2{
\maxnorm{U_{[1,n]}} &\equiv& \displaystyle\min_{U} & 
\maxnorm{ U }\cr
&&\text{s.t.}&U\ket{e_i} = \ket{b_i}  \:\: \text{for all }i=1,\ldots,n,\cr
&&&\text{with }U \in \myUgrp(r).&eqn-problem-original-min-U\cr
}
where $\ket{e_i}$ is the unit vector with $1$ at the $i$th entry and $0$ everywhere else.
Note that the notation $U_{[1,n]}$ means that the columns $1$ to $n$ of $U$ are fixed as in Eq.~\eqref{eqn-U-with-missing-columns}.
In other words,
$$
\maxnorm{ g(F_1,F_2,\dots,F_{\KrausDim}) } 
=
\maxnorm{U_{[1,n]}}.
$$

\subsection{Partial $U$ problem with one vector}
Consider a special case.
The ``partial $U$ problem'' \eqref{eqn-problem-original-min-U} with 
only one vector
has the following solution~\cite{Fung:2013:Time-energy}:

\dmathX2{
\maxnorm{U_{[i,i]}}
&\equiv& \min_{U} & \maxnorm{ U }\cr
&&\text{s.t.}&U\ket{e_i} = \ket{b_i}\text{ with }
U \in \myUgrp(r)&eqn-problem-min-U-single-vector\cr
&=&&\hspace{-18pt}\cos^{-1}\left[\operatorname{Re}(\braket{e_i}{b_i})\right].\cr
}
We remark the solution does not depend on the actual form of $\ket{e_i}$ and $\ket{b_i}$.
Note that the notation $U_{[i,i]}$ means that column $i$ of $U$ is fixed.

\subsection{Partial $U$ problem -- lower bound}

Since the feasible set of problem \eqref{eqn-problem-min-U-single-vector} contains that of problem \eqref{eqn-problem-original-min-U}, 
$$
\maxnorm{U_{[1,n]}}
\ge
\maxnorm{U_{[i,i]}} \text{ for all } i=1,\dots,n.
$$
Thus, a lower bound to the time-energy cost is
\begin{equation}
\label{eqn-maxnorm-lower-bound-diagonal1}
\maxnorm{U_{[1,n]}}
\ge
\max_{1\leq i \leq n}
\cos^{-1}\{\operatorname{Re}[U(i,i)]\}.
\end{equation}
where $\cos^{-1}$ always returns an angle in the range $[0,\pi]$.
Note that $\braket{e_i}{b_i}$ simply corresponds to the $i$th diagonal element of $U$.

Based on Eq.~\eqref{eqn-maxnorm-lower-bound-diagonal1}, two more bounds using the eigenvalues and singular values are derived:
\begin{align}
\label{eqn-maxnorm-lower-bound-diagonal2}
\maxnorm{U_{[1,n]}}
&\ge
\max_{1\leq i \leq n}
\cos^{-1}\{\operatorname{Re}[
\lambda_i(F_1^\text{top})
]\}, \text{ and}
\\
\label{eqn-maxnorm-lower-bound-diagonal3}
\maxnorm{U_{[1,n]}}
&\ge
\cos^{-1} \left[ \sigma_\text{min}(F_1) \right]
\end{align}
where $\lambda_i$ denotes the $i$th eigenvalue of its argument and $\sigma_\text{min}$ denotes the minimum singular value of its argument.
To get Eqs.~\eqref{eqn-maxnorm-lower-bound-diagonal2} and \eqref{eqn-maxnorm-lower-bound-diagonal3},
we need the following lemma.
\begin{lemma}
\label{lemma-U-norm-unchanged-by-conjugation}
{\myrm
(Lemma~1 in Ref.~\cite{Fung:2013:Time-energy})
$$
\maxnorm{U_{[1,n]}}
=
\maxnorm{(\tilde{Q} U \tilde{Q}^\dag)_{[1,n]}}
$$
for any unitary matrix 
$Q \in \myUgrp(\ChannelDim)$ with
\begin{equation*}
\tilde{Q}=
\begin{bmatrix}
Q & 0
\\
0& \bf{1}
\end{bmatrix}
\in \myUgrp(r) .
\end{equation*}
}
\end{lemma}
To get Eq.~\eqref{eqn-maxnorm-lower-bound-diagonal2},
we apply Schur decomposition to 
the first $n$ rows of $F_1$ (which is a square matrix denoted as $F_1^\text{top}$)
to obtain its eigenvalues on the diagonal of a triangular matrix and use Lemma~\ref{lemma-U-norm-unchanged-by-conjugation} to cancel out the left and right unitary matrices.
This triangular matrix becomes the new top-left block of $U$.
To obtain Eq.~\eqref{eqn-maxnorm-lower-bound-diagonal3},
we apply singular value decomposition to $F_1$ to get $F_1=V D Q$ ($V$ and $Q$ are unitary and $D$ is diagonal) and use Lemma~\ref{lemma-U-norm-unchanged-by-conjugation} to cancel out the right unitary matrix $Q$ giving the new $U(i,i)= (V D)(i,i)$.
Next, note that $\operatorname{Re}[(V D)(i,i)] \le D(i,i)$ since the magnitude of every element of $V$ (being unitary) is at most one.
Thus, Eq.~\eqref{eqn-maxnorm-lower-bound-diagonal3} is a looser bound than Eq.~\eqref{eqn-maxnorm-lower-bound-diagonal1}.

In general, we may take the maximum of the RHS of 
Eqs.~\eqref{eqn-maxnorm-lower-bound-diagonal1}-\eqref{eqn-maxnorm-lower-bound-diagonal3} to serve as the lower bound.

\subsection{Partial $U$ problem -- diagonal $F_1$}

An exact time-energy cost is obtained for a special case.
If the top-left $n \times n$ block of $U$ is diagonal (i.e., $F_1$ is diagonal if it is square),
we have 
\begin{align}
\label{eqn-maxnorm-special-case1}
\maxnorm{U_{[1,n]}}
=
\max_{1 \le i \le n}
\cos^{-1}\{\operatorname{Re}[U(i,i)]\}
\end{align}
(c.f. Eq.~(44) of Ref.~\cite{Fung:2013:Time-energy}).

In general, if $F_1$ is Hermitian, 
it can be diagonalized 
and, based on Lemma~\ref{lemma-U-norm-unchanged-by-conjugation},
Eq.~\eqref{eqn-maxnorm-special-case1} becomes
\begin{align}
\label{eqn-maxnorm-special-case2}
\maxnorm{U_{[1,n]}}
=
\cos^{-1}
\left[
\lambda_\text{min} (F_1)
\right],
\end{align}
where $\lambda_\text{min}$ denotes the minimum eigenvalue of its argument.

\bibliographystyle{apsrev4-1}

\bibliography{paperdb,paperdb_2}

\end{document}